\documentclass[11pt]{article}

\usepackage[latin1]{inputenc}
\usepackage{amsmath}
\usepackage{amssymb}
\usepackage{listings}
\usepackage{vaucanson-g}
\usepackage{wasysym}
\usepackage{epsfig}
\usepackage{ifthen}
\usepackage{color}
\usepackage{array}
\usepackage{longtable}
\usepackage{calc}
\usepackage{multirow}

\newtheorem{proof}{Proof}

\newtheorem{theorem}{Theorem}

\newcommand{\fado}{{\bf F{\cal A}do}}

\newcommand{\lixo}[1]{}
\newcommand{\tirar}[1]{}
\newcommand{\ICDFA}{\textsc{ICDFA}}
\newcommand{\ICDFAE}{$\mbox{\textsc{ICDFA}}_\emptyset$}
\newcommand{\DFAE}{$\mbox{\textsc{DFA}}_\emptyset$}
\newcommand{\DFA}{\textsc{DFA}}
\newcommand{\floor}[1]{\lfloor #1 \rfloor}

\begin{document}

\title{
Aspects of Enumeration and Generation with a String Automata
Representation
\thanks{Work partially
    funded by Funda\c{c}\~ao para a Ci\^encia e Tecnologia (FCT) and
    Program POSI.}\footnote{This paper was presented at the 8th
    Workshop on Descriptional  Complexity of Formal Systems, DCFS'06.}}

\author{
  Marco Almeida\\
{\tt mfa@ncc.up.pt}\\
\and Nelma Moreira \\
  {\tt nam@ncc.up.pt}\\
\and Rog\'erio Reis\\
  {\tt rvr@ncc.up.pt}\\
  DCC-FC \ \& LIACC,  Universidade do Porto \\
  R. do Campo Alegre 823, 4150 Porto, Portugal
}
\date{}
\maketitle

\lstdefinelanguage{algo}{
  morekeywords={for,if,then,do,while,return,break,else,continue,and,nil,to,
downto,print,def},
  morecomment=[l]{\%}}

\lstset{language=algo, 
  aboveskip=10pt, belowskip=10pt, 
  mathescape=true , basicstyle=\small}

\begin{abstract}In general, the representation of combinatorial
  objects is decisive for the feasibility of several enumerative
  tasks.  In this work, we show how a (unique) string representation
  for (complete) initially-connected deterministic automata (\ICDFA's)
  with $n$ states over an alphabet of $k$ symbols can be used for
  counting, exact enumeration, sampling and optimal coding, not only
  the set of ~\ICDFA's but, to some extent, the set of regular
  languages. An exact generation algorithm can be used to partition
  the set of ~\ICDFA's in order to parallelize the counting of minimal
  automata (and thus of regular languages).  We present also a uniform
  random generator for ~\ICDFA's that uses a table of pre-calculated
  values.  Based on the same table it is also possible to obtain an
  optimal coding for ~\ICDFA's.
\end{abstract}


\section{Introduction}

In general, the representation of combinatorial objects is decisive
for the feasibility of several enumerative tasks.  In this work, we
show how a (unique) string representation for (complete)
initially-connected deterministic automata (\ICDFA's) with $n$ states
over an alphabet of $k$ symbols can be used for counting, exact
enumeration, sampling and optimal coding, not only the set of
~\ICDFA's but, to some extent, the set of regular languages. The key
fact is that string representations are characterized by a set of
rules that allow an exact and ordered generation of all its elements.
An exact generation algorithm can be used to partition the set
of~\ICDFA's in order to parallelize the counting of minimal automata,
and thus of regular languages.  With the same set of rules it is
possible to design a uniform random generator for~\ICDFA's that uses a
table of pre-calculated values (as usual in combinatorial
decomposition approaches).  Based on the same table it is also
possible to obtain an optimal coding for~\ICDFA's (with or without
final states).

In the next section, some
definitions and notation are introduced. In Section~\ref{sec:str} we
review the string representation of non-isomorphic~\ICDFAE's
(i.e., \ICDFA's without final states), and how
it can be used to generate and enumerate all \ICDFA's. We also relate
those methods to the ones presented by Champarnaud and
Paranth\"oen in~\cite{champarnaud:_random_gener_dfas}, by giving a new
enumerative result. In Section~\ref{sec:gen}, we briefly describe the
implementation of a generator algorithm for \ICDFAE's.
Section~\ref{sec:count} presents the methods for parallelizing the
counting of languages by slicing the universe of~\ICDFAE's and some
experimental results are given. A uniform random generator for
\ICDFAE's is described in Section~\ref{sec:rand} along with some
experimental results and statistical tests. Using the recurrence
formulae defined in Section~\ref{sec:rand}, we show in
Section~\ref{sec:enum} how we can associate an integer with
an~\ICDFAE's and vice-versa. Section~\ref{sec:con} concludes with final
remarks.

\section{Preliminaries}
\label{sec:preliminares}
Given two integers $m<n$ we represent the set $\{i \in \mathbb{N}\mid
m\leq i \leq n\}$ by $[m,n]$.  A \emph{deterministic finite automaton}
(\DFA{}) ${\cal A}$ is a quintuple $(Q,\Sigma,\delta,q_0,F)$ where $Q$
is a finite set of states, $\Sigma$ the alphabet, i.e, a non-empty
finite set of symbols, $\delta: Q \times \Sigma \rightarrow Q$ is the
transition function, $q_0$ the initial state and $F\subseteq Q$ the
set of final states. The \emph{size of the automaton} is given by
$|Q|$. We assume that the transition function is total, so we consider
only \emph{complete} \DFA{}'s.
As we are not interested in the labels of the states, we can
represent them by an integer $i\in[0,|Q|-1]$.  The transition function
$\delta$ extends naturally to $\Sigma^\star$. \lixo{ for all $q\in Q$,
  if $x=\epsilon$ then $\delta(q,\epsilon)=q$; if $x=y\sigma$ then
  $\delta(q,x)=\delta(\delta(q,y),\sigma)$.}A \DFA{} is
\emph{initially-connected}\footnote{Also called \emph{accessible}.}
(\ICDFA{}) if for each state $q\in Q$ there exists a string $x\in
\Sigma^\star$ such that $\delta(q_0,x)=q$.  The \emph{structure} of an
automaton $(Q,\Sigma,\delta,q_0)$ denotes a \DFA{} without its final
state information and is referred to as a \DFAE.  For each structure,
there will be $2^n$ \DFA{}'s, if $|Q|=n$. We denote by \ICDFAE{} the
structure of an~\ICDFA{}. Two \DFA{}'s ${\cal
  A}=(Q,\Sigma,\delta,q_0,F)$ and ${\cal
  A}'=(Q',\Sigma,\delta',q_0',F')$ are called \emph{isomorphic} (by
states) if there exists a bijection $f:Q \rightarrow Q'$ such that
$f(q_0)=q_0'$ and for all $\sigma\in \Sigma$ and $q\in Q$,
$f(\delta(q,\sigma))=\delta'(f(q),\sigma)$.  Furthermore, for all
$q\in Q$, $q\in F$ if and only if $f(q)\in F'$.  The \emph{language}
accepted by a \DFA{} ${\cal A}$ is $L({\cal A})=\{x\in
\Sigma^\star\mid \delta(q_0,x)\in F\}$. Two \DFA{} are
\emph{equivalent} if they accept the same language.  Obviously, two
isomorphic automata are equivalent, but two non-isomorphic automata
may also be equivalent. A \DFA{}\ ${\cal A}$ is \textsl{minimal} if there
is no \DFA{}\ ${\cal A}'$ with fewer states equivalent to ${\cal A}$.
Trivially a minimal \DFA{} is an \ICDFA{}.  Minimal \DFA{}'s are
unique up to isomorphism.  Domaratzki et al.~\cite{domaratzki02} gave
some asymptotic estimates and explicit computations of the number of
distinct languages accepted by finite automata with $n$ states over an
alphabet of $k$ symbols.  Given $n$ and $k$, they denoted by $f_k(n)$
the number of pairwise non-isomorphic minimal \DFA's and by $g_k(n)$
the number of distinct languages accepted by \DFA's, where
$g_k(n)=\sum_{i=1}^n f_k(i).$

\section{Strings for \ICDFA's}
\label{sec:str}

Reis et al.~\cite{reis05:_repres_finit_dcfs} presented a unique
string representation for non-isomorphic \ICDFAE's. In this section,
we briefly review this representation and how it can be used to
generate and enumerate all \ICDFA s. We also give a new enumerative
result and relate this representation to the one presented by
Champarnaud and Paranth\"oen in~\cite{champarnaud:_random_gener_dfas}.

Given a complete \DFAE{} $(Q,\Sigma,\delta,q_0)$ with $|Q|=n$ and
$|\Sigma|=k$ , consider a total order $<$ over $\Sigma$. We can define a
canonical order over the set of the states by exploring the automaton
in a breadth-first way choosing at each node the outgoing edges in the
order considered for $\Sigma$. 
If we restrict this representation to \ICDFAE{}'s, then this
representation is unique and defines an order over the set of its
states. For instance, consider the following \ICDFAE{} and consider
the alphabetic order in $\{\mathtt{a},\mathtt{b},\mathtt{c}\}$.

\begin{center}
  \TinyPicture\VCDraw{
      \begin{VCPicture}{(0,2)(7,7)}
        \State[A]{(1,6)}{A}
        \State[C]{(6,6)}{C}
        \State[B]{(1,3)}{B}
        \State[D]{(6,3)}{D}
        \Initial{A}
         \LoopN{A}{\texttt{c}} 
         \ArcL{A}{C}{\texttt{a}}
         \ArcR{A}{B}{\texttt{b}}
         \ArcL{C}{A}{\texttt{c}}
         \ArcR{C}{D}{\texttt{b}}
         \EdgeL{C}{B}{\texttt{a}}
         \LoopE{D}{\texttt{b}} 
         \ArcL{D}{B}{\texttt{c}}
         \ArcR{D}{C}{\texttt{a}}
         \LoopW{B}{\texttt{c}} 
         \ArcR{B}{A}{\texttt{b}}
         \ArcL{B}{D}{\texttt{a}}
  \end{VCPicture}}
    \end{center}
    \noindent The states ordering is
    \texttt{A},\texttt{C},\texttt{B},\texttt{D} and
    $[1,2,0,2,3,0,3,0,2,1,3,2]$ is its string representation.
    Formally, let $\Sigma = \{\sigma_i\mid i\in [0,k-1]\}$, with
    $\sigma_0<\sigma_1<\cdots<\sigma_{k-1}$.  Given an \ICDFAE{}
    $(Q,\Sigma,\delta,q_0)$ with $|Q|=n$, the representing string is
    of the form $(s_i)_{i\in[0,kn-1]}$ with $ s_i \in [0,n-1]$ and
    $s_i=\delta(\floor{i/k},\sigma_{i\bmod{k}})$.

Let  $(s_i)_{i\in [0,kn-1]}$ with
$s_i\in[0,n-1]$ be a string satisfying the following conditions:
\begin{gather}
  (\forall m\in [2,n-1])(\forall i\in[0,kn-1])(s_i=m\;\Rightarrow\;(\exists
  j\in[0,i-1])\,s_j=m-1).\tag{\textbf{R1}}\label{eq:r1}\\
  (\forall m \in [1,n-1])(\exists
  j\in[0,km-1])\,s_j=m.\tag{\textbf{R2}}\label{eq:r2} 
\end{gather}

In~\cite{reis05:_repres_finit_dcfs} the following theorem was proved.

\begin{theorem}
  There is a one-to-one mapping between $(s_i)_{i\in[0,kn-1]}$ with
  $s_i \in [0, n-1]$ satisfying rules \ref{eq:r1} and \ref{eq:r2},
  and the non-isomorphic \ICDFAE{}'s with $n$ states, over an alphabet
  $\Sigma$ of size $k$.
\end{theorem}

We note that this string representation can be extended to
non-complete \ICDFAE{}'s, by representing all missing transitions with
the value $-1$. In this case, rules \ref{eq:r1} and \ref{eq:r2} remain
valid, and we can assume that the transitions from this state are into
itself. However for enumeration and generation purposes we do not
consider non-complete \ICDFAE{}'s.

In order to have an algorithm for the enumeration and generation of
\ICDFAE's, instead of rules \ref{eq:r1} and \ref{eq:r2} an alternative
set of rules were used.  
For $n=1$ there is only one (non-isomorphic) \ICDFAE{} for each
$k\geq 1 $, so we assume in the following that $n>1$.
In a string representing an \ICDFAE{}, let
$(f_j)_{j\in[1,n-1]}$ be the sequence of indexes of the first
occurrence of each state label~$j$. For explanation purposes, we call
those indexes \emph{flags}.  
It is easy to see that  (\ref{eq:r1}) and  (\ref{eq:r2}) correspond
respectively to (\ref{eq:g1}) and (\ref{eq:g2}):
  \begin{gather}
    (\forall j\in [2,n-1])(f_j>f_{j-1});\tag{\textbf{G1}}\label{eq:g1}\\
    (\forall m \in  [1,n-1])\; (f_m< km).\tag{\textbf{G2}}\label{eq:g2}
  \end{gather}
This means that $f_1\in [0,k-1]$, and $f_{j-1}< f_{j}< kj$ for
$j\in[2,n-1]$. We begin by counting the number of sequences of flags allowed.
\begin{theorem}\label{thm:cat}
  Given $k$ and $n$, the number of sequences $(f_j)_{j\in[1,n-1]}$,
  ${F_{k,n}}$, is given by
    $$
  F_{k,n}= \sum_{f_1=0}^{k-1} \sum_{f_2=f_1+1}^{2k-1}
    \cdots \sum_{f_{n-1}=f_{n-2}+1}^{k(n-1)-1}
  1= \binom{kn}{n}\frac{1}{(k-1)n+1}=C_n^{(k)};$$
  where $C_n^{(k)}$ are the (generalised) Fuss-Catalan numbers.
\end{theorem}
\begin{proof}
  The first equality follows directly from the definition of the
  $(f_j)_{j\in[1,n-1]}$. For the second, note that $C_n^{(k)}$
  enumerates $k$-ary trees with $n$ internal nodes, ${\cal T}_n^k$
  (see for instance~\cite{sedgewick96:_analy_algor}).  In particular,
  for $k=2$, $C_n^2$ are exactly the Catalan numbers that count binary
  trees with $n$ internal nodes. This sequence appears in
  Sloane~\cite{sloane03:_encyc_integ_sequen} as \textbf{A00108} and
  for $k=3$ and $k=4$ as \textbf{A001764} and \textbf{A002293}
  sequences, respectively. So it suffices to give a bijection between
  these trees and the sequences of flags.  Recall that a $k$-ary tree
  is an external node or an internal node attached to an ordered
  sequence of $k$, $k$-ary sub-trees.  

  \begin{figure}[h]
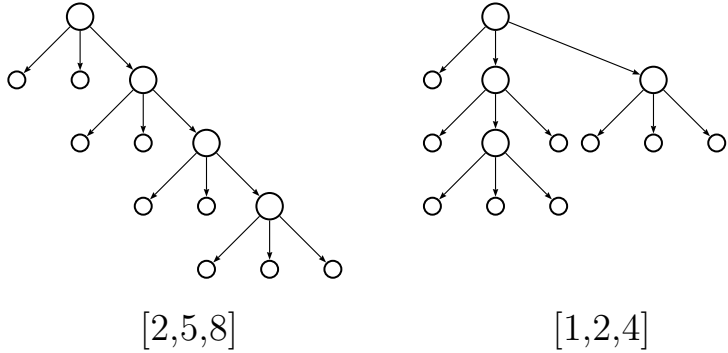
 
    \begin{center}
    \begin{tabular}{ccccc}
     \TinyPicture\VCDraw{
        \begin{VCPicture}{(-1, -10)(8, 0)}
          \MediumState \State{(0, -1)}{R} \SmallState \State{(-2, -3)}{A}
          \State{(0, -3)}{B} \MediumState \State{(2, -3)}{C} \Edge{R}{A}
          \Edge{R}{B} \Edge{R}{C} \SmallState \State{(0, -5)}{D} \State{(2,
            -5)}{E} \MediumState \State{(4, -5)}{F} \Edge{C}{D} \Edge{C}{E}
          \Edge{C}{F} \SmallState \State{(2, -7)}{G} \State{(4, -7)}{H}
          \MediumState \State{(6, -7)}{I} \Edge{F}{G} \Edge{F}{H} \Edge{F}{I}
          \SmallState \State{(4, -9)}{J} \State{(6, -9)}{K} \State{(8, -9)}{L}
          \Edge{I}{J} \Edge{I}{K} \Edge{I}{L}
        \end{VCPicture}}
&&&&
     \TinyPicture\VCDraw{
\begin{VCPicture}{(-1,-10)(8, 0)}
          \MediumState \State{(0, -1)}{R}
          \SmallState \State{(-2, -3)}{A}
          \MediumState \State{(0, -3)}{B} 
          \State{(5, -3)}{C} \Edge{R}{A}
          \Edge{R}{B} \Edge{R}{C} 

          \SmallState \State{(-2, -5)}{D} 
          \MediumState \State{(0,-5)}{E} \SmallState \State{(2, -5)}{F} 
          \Edge{B}{D} 
          \Edge{B}{E}
          \Edge{B}{F} 
          \SmallState \State{(3, -5)}{G} \State{(5, -5)}{H}
          \State{(7, -5)}{I} 
          \Edge{C}{G} \Edge{C}{H} \Edge{C}{I}
          \SmallState \State{(-2, -7)}{J} \State{(0, -7)}{K} \State{(2, -7)}{L}
          \Edge{E}{J} \Edge{E}{K} \Edge{E}{L}
        \end{VCPicture}

}\\
{\Large[2,5,8]}&&&&{\Large[1,2,4]}
\vspace{-0.5cm}
    \end{tabular}
    \end{center}
    \caption{Two $3$-ary trees with 4 internal nodes and the correspondent 
       sequence of flags.}
    \label{fig:tree}
  \end{figure}
  Let ${\cal T}_n^k$ be a $k$-ary tree and let $<$ be a total order
  over $\Sigma$. For each internal node $i$ of ${\cal T}_n^k$ its
  outgoing edges can be ordered left-to-right and attached a unique
  symbol of $\Sigma$ according to $<$.  Considering a breadth-first,
  left-to-right, traversal of the tree and ignoring the root node
  (that is considered the $0$-th internal node), we can represent
  ${\cal T}_n^k$, uniquely, by a bitmap where a $0$ represents an
  external node and a $1$ represents an internal node.  As the number
  of external nodes are $(k-1)n+1$, the length of the bitmap is $kn$.
  Moreover the $j+1$-th block of $k$ bits corresponds to the children
  of the $j$-th internal node visited, for $j\in [0,n-1]$. For
  example, the bitmaps of the trees in Figure~\ref{fig:tree} are
  $[0,0,1,0,0,1,0,0,1,0,0,0]$ and $[0,1,1,0,1,0,0,0,0,0,0,0]$,
  respectively. The positions of the $1$'s in the bitmaps correspond
  to a sequence of flags, $(f_i)_{i\in[1,n-1]}$, i.e., $f_i$
  corresponds to the number of nodes visited before the $i$-th
  internal node (excluding the root node). It is obvious that
  $(f_i)_{i\in[1,n-1]}$ verify~\ref{eq:g1}.  For~\ref{eq:g2}, note
  that for the each internal node the outdegree of the previous internal
  nodes is $k$. Conversely, given a sequence of flags $(f_j)_{j\in[1,n-1]}$, we
  construct the bitmap such that $b_{f_i}$=1 for $i \in [1,n-1]$ and
  $b_j=0$ for the remaining values, for $j\in[0,kn-1]$. As above, for
  the representation of the $j+1$-th internal node, $\floor{f_j/k}$
  gives the parent and ${f_j\bmod k}$ gives its position between its
  siblings (in breadth-first, left-to-right traversal).
\end{proof}
To generate all the \ICDFAE's, for each allowed sequence of flags
$(f_j)_{j\in[1,n-1]}$,  all the remaining symbols $s_i$
can be generated according to the following
rules:
\begin{gather}
  i<f_1\;\Rightarrow\;s_i=0;\tag{\textbf{G3}}\label{eq:g3}\\
  (\forall j \in [1,n-2]) (f_j<i<f_{j+1} \;\Rightarrow s_i\in [0,j]);
  \tag{\textbf{G4}}\label{eq:g4}\\
  i>f_{n-1}\;\Rightarrow\; s_i\in [0,n-1].\tag{\textbf{G5}}\label{eq:g5}
\end{gather}

In~\cite{reis05:_repres_finit_dcfs} a simple combinatorial argument
was given to show that
\begin{theorem}
  The number of strings $(s_i)_{i\in[0,kn-1]}$ representing \ICDFAE{}'s with
  $n$ states over an alphabet of $k$ symbols is given by
  \begin{equation}
    \label{eq:Bk}
    B_{k,n}= \sum_{f_1=0}^{k-1} \sum_{f_2=f_1+1}^{2k-1}
    \sum_{f_3=f_2+1}^{3k-1} \cdots \sum_{f_{n-1}=f_{n-2}+1}^{k(n-1)-1}
    \prod_{i=2}^{n} i^{f_i-f_{i-1}-1};
  \end{equation}
  where $f_n=kn$.
\end{theorem}

In Section~\ref{sec:rand} we give other recursive definition that
is more adequate for tabulation.

\subsection{Analysis of the Champarnaud et al. Method}

Champarnaud and
Parantho\"en in~\cite{champarnaud:_random_gener_dfas,paranthoen04},
generalizing work of Nicaud~\cite{nicaud00}
presented a method to generate and enumerate \ICDFAE's, although not
giving an explicit and compact representation for them, as the string
representation used here.  An order $<$ over $\Sigma^\star$ is a
\emph{prefix order} if $(\forall x \in \Sigma^\star)(\forall \sigma\in
\Sigma) x < x\sigma$.  Let ${\cal A}$ be an~\ICDFAE{} over $\Sigma$
with $k$ symbols and $n$ states. Given a prefix order in
$\Sigma^\star$, each automaton state is ordered according to the first
word $x\in \Sigma^\star$ that reaches it in a simple path from the
initial state. The sets of this words $\{\cal P\}$  are in bijection with $k$-ary
trees with $n$ internal nodes, and therefore to the set of sequences
of flags, in our representation\footnote{Indeed our order on the
  states induces a prefix order in $\Sigma^\star$.}. Then it is possible
to obtain a valid \ICDFAE{} by adding other transitions in a way that preserves
the previous state labelling. For the generation of the sets ${\cal
  P}$ it is used another set of objects that are in bijection with
  $k$-ary trees with $n$ internal nodes and are called generalised tuples.
The number of
 \ICDFAE{}'s is  computed using recursive formulae associated with
 generalized tuples, akin the ones we present in
 Section~\ref{sec:rand}.
\section{Generating  \ICDFAE's}
\label{sec:gen}

In this section, we present a method to generate all \ICDFAE's, given
$k$ and $n$. We start with an initial string, and then consecutively
iterate over all allowed strings until the last one is reached. The
main procedure is the one that given a string returns the
\emph{next} legal one. For each $k$ and $n$, the first \ICDFAE{} is
represented by the string $0^{k-1}10^{k-1}\ldots(n-1)0^{k}$ and the
last is represented by $12\ldots(n-1)(n-1)^{(k-1)n+1}$.  According to
the rules~\ref{eq:g1}-~\ref{eq:g5}, we first generate a sequence of
flags, and then, for each one, the set of strings representing the \ICDFAE's
in lexicographic order. The  algorithm to generate the next sequence of
flags is the following, where the initial sequence
of flags is $(ki-1)_{i\in[1,n-1]}$:
\begin{lstlisting}
def $\; \mathbf{nextflags}(i)$:
  if $i$==$1$ then $f_i$ =  $f_i$ - 1
  else 
     if ($f_i$-1 == $f_{i-1}$) then  
        $f_i$ = $k*i-1$
        $\mathbf{nextflags}(i-1)$
     else $f_{i}=f_{i}-1$
\end{lstlisting}
To generate a new sequence, we must call~\textbf{nextflags(n-1)}.
Given the rules~\ref{eq:g1} and~\ref{eq:g2} the correctness of the
algorithm is easily proved. When a new sequence of flags is generated,
the first \ICDFAE{} is represented by a string with $0$s in all other
positions (i.e., the lower bounds in rules~\ref{eq:g3}--\ref{eq:g5}).
The following  strings, with the same sequence of flags, are computed
lexicographically using the procedure \textbf{nexticdfa}, called with $a=n-1$ and
$b=k-1$:
\begin{lstlisting}
def $\mathbf{nexticdfa}(a,b)$:
   $i$ = $a*k + b$
        if $a$ < $n-1$ then
            while $i \in (f_j)_{j\in[1,n-1]}$:
               for $k=i+1$ to $kn-1$:
                   if $k \notin (f_j)_{j\in[1,n-1]}$ then $s_k=0$
               $b$ = $b$ - $1$
               $i$ = $i$ - $1$
        $f_j$ = the nearest flag not exceeding $i$       
        if $s_i == s_{f_j}$ then 
           $s_i= 0$
           if $b == 0$ then $\mathbf{nexticdfa}(a-1, k-1)$
           else $\mathbf{nexticdfa}(a, b-1)$
        else $s_i$ = $s_i$ + 1 
\end{lstlisting}
 
Note that the last string for each sequence of flags has the
value $s_l=j$ for $l\in[f_j+1,f_{j+1}-1]$, with $j\in[1,n-1]$. The
time complexity of the generator is linear in the number of automata.
As an example, for $k=2$ and $n=9$ it took about $12$ hours to
generate all the $705068085303$ \ICDFAE's, using a AMD Athlon at
2.5GHz. Finally, for the generation of~\ICDFA's
we only need to add to the string representation of an~\ICDFAE, a
string of $n$ $0$'s and $1$'s, correspondent to one of the $2^n$
possible choices of final states.

\section{Counting Regular Languages (in Slices)}
\label{sec:count}
To obtain the number of languages accepted by \DFA's with $n$ states
over an alphabet of $k$ symbols, we can generate all~\ICDFA's,
determine which of them are minimal ($f_k(n)$) and calculate the value
of
$g_k(n)$.
 Obviously, this is in general an intractable procedure.  But for
 small values of $n$ and $k$ some experiments can take place.  We must
 have an efficient implementation of a minimization algorithm, not
 because of the size of each automaton but because the number of
 automata we need to cope with. For that we implemented Hopcroft's
 minimization algorithm~~\cite{hopcroft71}, using efficient set
 representations. For very small values of $n$ and $k$ ($n+k<16$) we
 represented sets as bitmaps and for larger values,
 AVL trees~\cite{gnu_libav_binar_searc_trees_librar} were used.

 The problem can be parallelized providing that the space search can
 be safely partitioned. Using the method presented in
 Section~\ref{sec:gen}, we can easily generate \emph{slices} of
 ~\ICDFAE's and feed them to the minimization algorithm. A
 \emph{slice} is a sequence of ~\ICDFAE's and is defined by a pair
 $(start,last)$, where $start$ is the first automaton in the sequence
 and $last$ is the last one.
If we have a set of CPUs available, each one can receive a slice,
generate all~\ICDFAE's (in that slice), generate all the
necessary~\ICDFA's and feed them to the
minimization algorithm.  For the generation of \ICDFA's, we used
the observation by Domaratzki \textit{et al.}~\cite{domaratzki02}, that is enough to test $2^{n-1}$
sets of final states, using the fact that a \DFA{} is minimal
\textit{iff} its complementary automaton is minimal too.
In this way, we can safely divide the search space and distribute each
slice to a different CPU.  Note that this approach relies in the
assumption that we have a much more efficient way to partition the
search space than to actually perform the search (in this case a
minimization algorithm).  The task of creating the slices
can be taken by a central process that successively generates the next
slice and at the end assembles all the results.  The \emph{server} can
run interactively with its \emph{slaves}, or it can generate all the
slices at once to be used later.  The server generates a \emph{slice}
using the generator algorithm presented in Section~\ref{sec:gen}. 
For this experiment we used two approaches. We developed a simple
\textnormal{slave management system} -- called \texttt{Hydra} ---
based on \texttt{Python} threads, that was composed by a server and a
variable set of \emph{slaves}. In this case, the slaves can be any
computer\footnote{We used all the normal desktop computers of our  colleagues in
  the CS Department.}. For each slice a process was executed via
\texttt{ssh}, and the result was returned to the server.  Another
approach was to use a computer grid, in particular 24 AMD Opteron 250 2.4GHz (dual
core).

\subsection{Experimental results}
\label{sec:expr}
In Table~\ref{tab:minimal}, we summarise some experimental
results.  Most of the values for $k=2$ and $k=3$, were already given
by Domaratzki {\em et al.} in~\cite{domaratzki02} and  the new
results are in bold in the table.
For $k=2, n=8$ we have
divided the universe of \ICDFAE's in 254 slices and the estimated CPU
time for each one to be processed is 11 days.
\begin{table}
  \centering
{\scriptsize\begin{tabular}{|c|c|l|l|l|l|l|} \hline
&n&\ICDFAE&\ICDFA&Minimal ($f_k(n)$)&Minimal \%&Time (s)\\\hline
$k=2$&2&12&48&24&50\%&0
\\&3&216&1728&1028&59\%&0.018
\\&4&5248&83968&56014&66\%&0.99
\\&5&160675&5141600&3705306&72\%&79.12
\\&6&5931540&379618560&286717796&75\%&8700
\\&7&256182290&32791333120&\textbf{25493886852}&77\%&1237313
\\\hline$k=3$&2&56&224&112&50\%&0.002
\\&3&7965&63720&41928&65\%&0.7
\\&4&2128064&34049024&26617614&78\%&494.72
\\&5&914929500&29277744000&\textbf{25184560134}&86\%&652703
\\\hline$k=4 $&2&240&960&480&50\%&0.01
\\&3&243000&1944000&1352732&69\%&23.5
\\&4&642959360&10287349760&\textbf{7756763336}&75\%&184808
\\\hline$k=5$&2&992&3968&1984&50\%&0.041
\\&3&6903873&55230984&36818904&66\%&756.2\\\hline
\end{tabular}}
   \caption{Performance and number of minimal automata.}
  \label{tab:minimal}
\end{table}

Moreover, the slicing process can give new insights about the
distribution of minimal automata.  Figure~\ref{fig:minimalp} presents
two examples of the values obtained for the rate of minimal
\DFA{}'s. For $n=7$ and $k=2$ we give the percentage of minimal
automata for each of the 257 slices we had used to divide the search
space (32791333120 \ICDFAE's). Each slice had about 100000 \ICDFAE's,
and so 128000000~\ICDFA's, and it took about 78 minutes to conclude
the process. The whole set of automata was
processed in 12 hours of real time of a CPU grid, that corresponds to
344 hours of CPU time.
\begin{figure}[h]
  \begin{center}
    \begin{tabular}[c]{cc}
    \includegraphics[angle=-90,width=7cm]{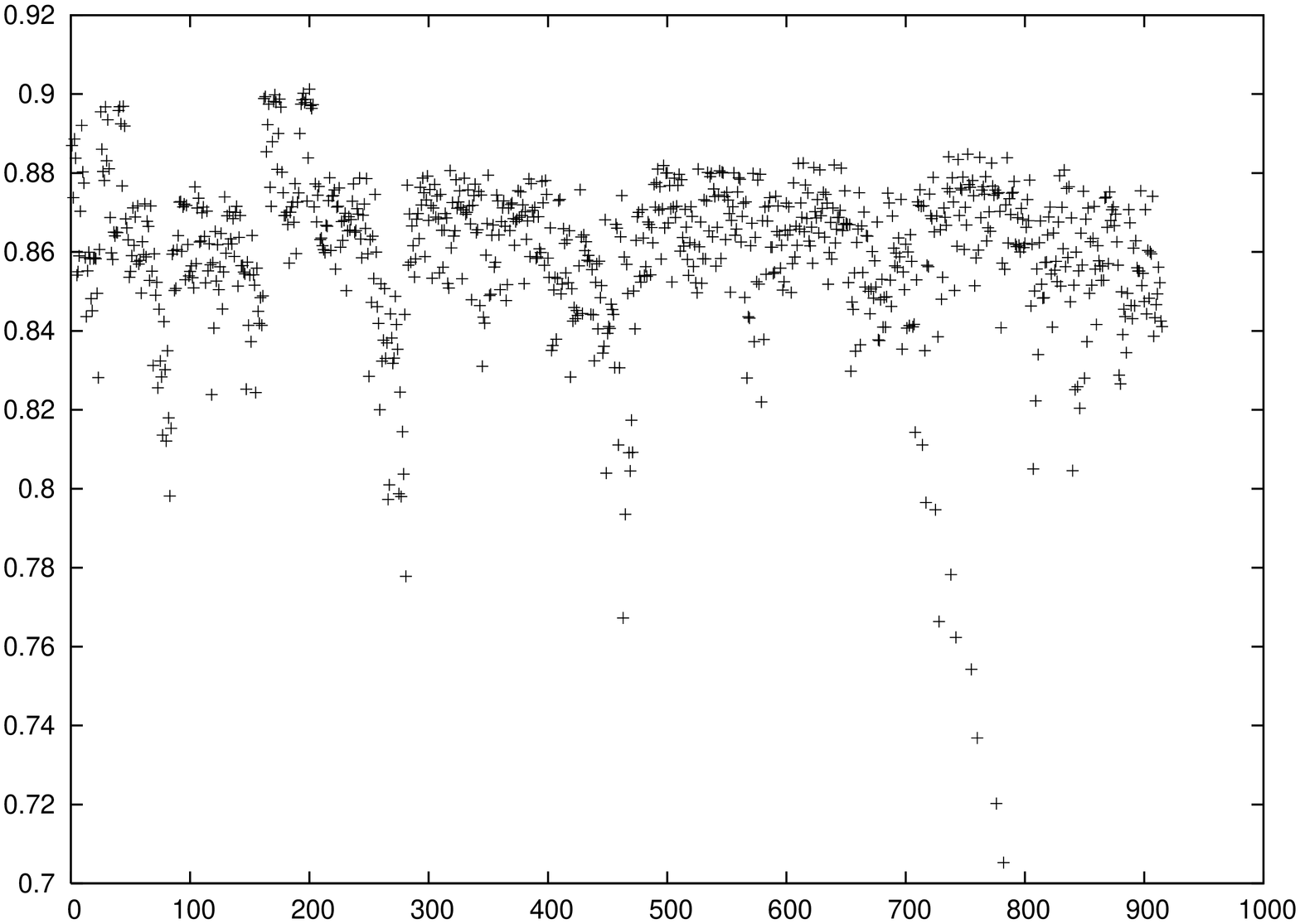}&
    \includegraphics[angle=-90,width=7cm]{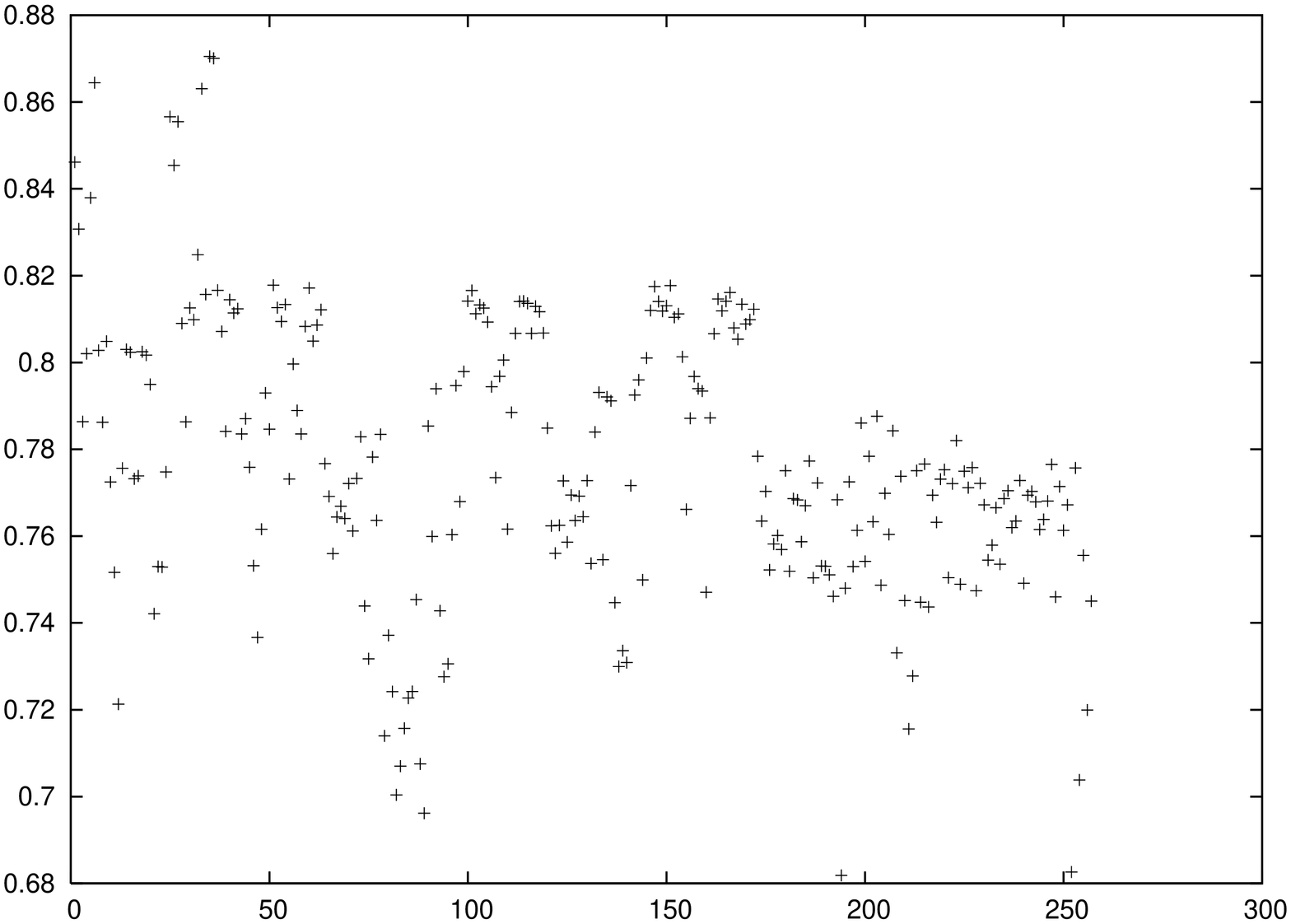}
    \end{tabular}
    \caption{Rate of minimal \DFA{}'s with ($k=3$,$n=5$) for $915$ slices
      and with ($k=2$,$n=7$) for $257$ slices.}
    \label{fig:minimalp}
  \end{center}
\end{figure}

\section{A Uniform Random Generator}
\label{sec:rand}
The \ICDFAE{} representation presented (Section \ref{sec:str}) permits
an easy random generation for \ICDFA s, and thus for \DFA{}s. To
randomly generate a \DFA{} for a given $n$ and $k$, it is only
necessary to: (i) randomly generate a valid sequence of flags
$(f_i)_{i\in[1,n-1]}$ according to \ref{eq:g1} and \ref{eq:g2}; (ii)
followed by the random generation of the rest of the $nk$ elements of
the string following \ref{eq:g3}--\ref{eq:g5} rules; (iii) and finally
the random generation of the set of final states. The uniformity issue
for steps (ii) and (iii) is quite straightforward.  For step (iii) it
is just necessary to use a uniform random integer generator for a
value $i\in [0,2^n]$. It is enough, for step (ii) the repeated use of
the same number generator for values in the range $[0,i]$ for $0\leq
i<n$ according to rules \ref{eq:g3}--\ref{eq:g5}. Step (i) is the only step
that needs special care.  Consider the case $n=5$ and $k=2$. Because
of rule \ref{eq:r1} flag $f_1$ can only be on positions $0$ or
$1$. But there are $140450$ \ICDFAE{}'s with $f_1$ in the first case
and only $20225$ in the second. Thus the random generation of flags,
to be uniform, must take this into account by making the first case more
probable than the second.  We can generate a random \ICDFAE{}
generating its representing string from left to right. Supposing that
flag $f_{m-1}$ is already placed at position $i$ and all the symbols
to its left are generated, i.e., the prefix $s_0s_1\cdots s_i$
is already defined, then the process can be described by:
\begin{lstlisting}
$r = \mathbf{random}(1,\sum\limits_{j=i+1}^{mk-1}N_{m,j})$
for $j=i+1$ to $mk-1$:
  if $r \in \left[\sum\limits_{l=i}^{j-1}N_{m,l},
    \sum\limits_{l=i}^{j}N_{m,l}\right]$ then return $i$
\end{lstlisting}
\noindent where \texttt{random(a,b)} is an uniform random generated
integer between \texttt{a} and \texttt{b}, and $N_{m,j}$ is the number
of \ICDFAE{}s with prefix $s_0s_1\cdots s_i$ with the first occurrence
of symbol $m$ in position $j$, making $N_{m,i}=0$ to simplify the
expressions. The values for $N_{m,j}$ could be obtained from expressions similar to
Equation~(\ref{eq:Bk}), and used in a program. But the program would
have a exponential time complexity. By expressing $N_{m,j}$ in a
recursive form, we  have, given $k$ and  $n$
\begin{equation}\label{eq:nij}
 \begin{array}{rcll}
  N_{n-1,j} &=& n^{nk-1-j}&\text{\;\;with\;}j\in[n-2,(n-1)k-1];\\
  N_{m,j} &=& \sum\limits_{i=0}^{(m+1)k-j-2}(m+1)^iN_{m+1,j+i+1}&\text{\;\;with\;}
  m\in[1,n-2],\\
  &&&j\in[m-1,mk-1].
\end{array}
\end{equation}
This evidences the fact that we keep repeating the same
computations with very small variations, and thus, if we use some kind
of tabulation of this values ($N_{m,j}$), with the
obvious price of memory space, we can create a version of a uniform
random generator, that apart of a constant overhead used for tabulation
of the function refered, has a complexity of
$\mathcal{O}(n^3k)\mathcal{O}(\mathtt{random})$. The algorithm is
described by the following:
\vspace{.2cm}

\begin{tabular}[c]{l|l}
\begin{lstlisting}
for $i=(n-1)k-1$ downto $n-2$:
    $N_{n-1,i} = n^{nk-1-i}$
for $m=n-2$ downto $1$:
    $N_{m,mk+1} = \sum\limits_{i=0}^{k-1} (m+1)^iN_{m+1,mk+i}$
    for $i=mk-2$ downto $m-1$:
        $N_{m,i} = (m+1)N_{m,i+1}+N_{m+1,i+1}$
$g = -1$
for $i=1$ to $n-1$:
    $f = \mathbf{generateflag}(i,g+1)$
    for $j=g+1$ to $f-1$:
        print $\mathbf{random}(0,i-1)$
     print $i$
     $g=f$
\end{lstlisting}&
\begin{lstlisting}
def $\mathbf{generateflag}(m,l)$:
    $r = \mathbf{random}$(0,$\sum\limits_{i=l}^{mk-1} m^{i-l}N_{m,i}$)
    for $i=l$ to $mk-1$:
        if $r < m^{i-l}N_{m,i}$ 
        then return $i$
        else $r = r - m^{i-l}N_{m,i}$
\end{lstlisting}
\end{tabular}\vspace{.2cm}

\noindent This means that with the same  AMD Athlon 64 at 2.5GHz, using a C
implementation with \texttt{libgmp}
\cite{gnu_multi_precis_arith_librar} the times reported in Table
\ref{tab:tempsrnd} were observed.
\begin{table}[ht]
  \centering
{\scriptsize  \begin{tabular}{|c|c|c|c|c|c|}
    \hline &$k=2$&$k=3$&$k=5$&$k=10$&$k=15$\\\hline
    $n=10$ &$0.10$s&$0.16$s&$0.29$s&$0.61$s&$1.30$s\\
    $n=20$ &$0.31$s&$0.49$s&$1.26$s&$4.90$s&$12.24$s\\
    $n=30$ &$0.54$s&$1.37$s&$3.19$s&$19.91$s&$62.12$\\
    $n=50$ &$1.61$s&$3.86$s&$17.58$s&$2.22$m&$947.71$s\\
    $n=75$ &$3.96$s&$12.98$s&$76.69$s&$700.20$s&$2459.34$s\\
    $n=100$ &$7.92$s&$36.33$s&$215.32$s&$2219.04$s&$8091.30$s\\
    \hline
  \end{tabular}}\vspace{.2cm}
  \caption{Times for the random generation of $10000$ automata.}
  \label{tab:tempsrnd}
\end{table}
It is possible, without unreasonable amounts of RAM to generate
random automata for unusually large values of $n$ and $k$. For example,
with $n=1000$ and $k=2$ the memory necessary is less than $450$MB. The
amount of memory used is so large not only because of the amount of
tabulated values, but because the size of the values is enormous. To
understand that, it is enough to note that the total number of
\ICDFAE{}'s for these values of $n$ and $k$ is greater than
$10^{3350}$, and the values tabulated are only bounded by this number.

\subsection{Statistical test of the random generator}
\label{sec:stat-test-rand}

Although the method used to generate random automata is, by its own
construction, uniform, we used $\chi^2$ test to evaluate the random generation
quality. The universe of \ICDFAE{}'s with $6$ states and $2$
symbols has a total size of $5931540$. This size is large enough for a
test with some significance and it is still reasonable, both in time
and space, to  perform the test. We generated three different sets of
$3000000$ \ICDFAE{}'s and perform the test in each one. Because of
the size of the data, we could not find  any tabulated values for
acceptance, and thus the following formula was used with $v=30000000-1$
and $x_p$ being the significance level ($1\%$ in this case):
$$v + 2 \sqrt{vx_p} + \frac{3}{4}x_p^2 - \frac{2}{3}.$$
The size of the data sets and the repetition of the test for three
times, is the recommended procedure by Knuth (\cite{knuth81:_art_comput_progr2},
pages 35--39). For the three experiments the values obtained 
were, respectively,  $5933268.92456$, $5925676.75108$ and $5935733.28172$,
  that are all smaller than the acceptance limit, that for this case
  was $5938980.75468$.
  
\section{Enumeration of \ICDFAE's}
\label{sec:enum}
In this section, we show how, given a string representation of an
\ICDFAE's of size $n$ over an alphabet of $k$ symbols, we can compute
its number in the generation order (described in
Section~\ref{sec:gen}) and vice-versa, i.e., given a number less than
$B_{k,n}$, we obtain the corresponding \ICDFAE. This provides an
optimal encoding for \ICDFAE's, as defined by M. Lothaire in
\cite{lothaire05:_applied_combin_words}, Chapter 9.  This bijection is
accomplished using the tables defined in Section~\ref{sec:rand} that
correspond to partial sums of Equation~(\ref{eq:Bk}).

\begin{theorem} 
  $B_{k,n}=\sum\limits_{l=0}^{k-1}N_{1,l}$.
\end{theorem}
\begin{proof}
  The result follows easily by expanding $N_{m,j}$ using
  Equations~(\ref{eq:nij}) and Equation~(\ref{eq:Bk}).
\end{proof}

\subsection{From \ICDFAE's to Integers}
\label{sec:in}

Let $(s_i)_{i\in[0,kn-1]}$ be an \ICDFAE's string representation, and
let $(f_j)_{j\in[1,n-1]}$ be the corresponding sequence of flags. From
the sequence of flags we obtain the following number, $n_f$,
\begin{equation}
  \label{eq:nf}
  n_f\;=\;\sum\limits_{i=1}^{n-1}\sum\limits_{j=f_i+1}^{ik-1}(
    i^{j-f_{i}}
N_{i,j}(\prod\limits_{m=1}^{i-1}(m^{f_{m+1}-f_{m}-1}))
\end{equation}
\noindent which  is the number of the first \ICDFAE{}\  with flags
$(f_j)_{j\in[1,n-1]}$.  Now we must add  the information provided by
the rest of the elements of the string
$(s_i)_{i\in[0,kn-1]}$:
\begin{equation}
  \label{eq:nr}
   n_r\; = \;\sum\limits_{j=1}^{n-1}\left(\sum\limits_{l=f_j+1}^{f_{j+1}-1}
s_l(j+1)^{f_{j+1}-1-l}\left(\prod\limits_{m=j+1}^{n-1}(m+1)^{f_{m+1}-f_{m}-1}\right)\right)
\end{equation}

And the corresponding number is $n_s=n_f+n_r$.

\subsection{From Integers to \ICDFAE's}
\label{sec:ni}

Given an integer $0\leq m<B_{k,n}$ a string representing uniquely an
\ICDFAE{} can be obtained using a method inverse of the one in the
last section.  The flags $(f_j)_{j\in[1,n-1]}$ are generated from
right-to-left, by successive subtractions.  The rest of the string
$(s_i)_{i\in[0,kn-1]}$ is generate considering the remainders of integer
divisions. The algorithms are the following: 

\begin{tabular}[c]{c|c}
\begin{lstlisting}{ frame=rtbl, frameround=tttt}
$s = 1$
for $i$ = $1$ to $n-1$:
    $j=i*k-1$
    $p=i^{j-f_{i-1}-1}$
    while $j>=i-1$ and $m\geq p*s*N_{i,j}$:
          $m$ = $m$ - $N_{i,j}*p*s$
          $j = j -1$
          $p = p/i$
    $s = s*i^{j-f_{i-1}-1}$
    $f_i = j$
\end{lstlisting}
&
\begin{lstlisting}
$i=k*n-1$
$j=n-1$
while $m>0$ and $j > 0$:
   while $m>0$ and $i >f_j$:
      $s_i= m \mod (j+1)$
      $m = m \div (j+1)$
      $i = i-1$ 
   $i = i-1$ 
   $j = j-1$ 
\end{lstlisting}
\end{tabular}

\section{Final Remarks}
\label{sec:con}
The methods here presented were implemented and tested to obtain both
exact and approximate values for the density of minimal automata.
Champarnaud \emph{et al.} in \cite{champarnaud:_random_gener_dfas},
checked a conjecture of Nicaud that for $k=2$ the number of minimal
\ICDFA's is about $80\%$ of the total, by sampling automata with $100$
states (for all possible number of final states).  Our results also
corroborate that conjecture, being the exact values for some small
values of $n$ and samples for greater values.  In particular, for $k=2$
and $n=100$ we obtained the same results as Champarnaud \emph{et al.}.
It seems that for $k>2$ almost all \ICDFA{}'s are minimal. For $k=3,5$
and $n=100$ that was also checked by Champarnaud \emph{et al.}. For a
confidence interval of $99$\% and significance level of $1$\% the
following table presents the percentages of minimal \ICDFA{}'s for
several values of $k$ and $n$, and each possible number of final
states. 

\begin{center}
{\scriptsize
\begin{tabular}{|c|c|c|c|c|c|c|c|c|c|c|c|}\hline
$k \backslash n$&$5$&$6$&$7$&$8$&$9$&$10$&$20$&$40$&$80$&$160$\\\hline
 3 & $85.8$\% & $90.8$\% & $93.3$\% & $95.0$\% & $96.1$\% & $96.7$\% & $98.7$\% & $99.4$\% & $99.7$\% &$99.8$\%\\\hline
 5 & $93.0$\% & $96.5$\% & $98.2$\% & $99.1$\% & $99.5$\% & $99.8$\% & $100.0$\% & $100.0$\% & $100.0$\% &$100.0$\%\\\hline
 7 & $93.7$\% & $96.8$\% & $98.4$\% & $99.2$\% & $99.6$\% & $99.8$\% & $100.0$\% & $100.0$\% & $100.0$\%&-- \\\hline
 9 & $93.7$\% & $96.9$\% & $98.4$\% & $99.2$\% & $99.6$\% & $99.8$\% &
$100.0$\% & $100.0$\% & -- &--\\\hline
11 & $93.8$\% & $96.9$\% & $98.4$\% & $99.2$\% & $99.6$\% & $99.8$\% &
$100.0$\% & $100.0$\% & --& -- \\\hline
13 & $93.7$\% & $96.9$\% & $98.4$\% & $99.2$\% & $99.6$\% & $99.8$\% & $100.0$\% & $100.0$\% & --& -- \\\hline
\end{tabular}}
\vspace{.2cm}  
\end{center}

A web interface to the random generator can be found in the
\fado\ project web page~\cite{fado}. Bassino and Nicaud
in~\cite{bassino:_enumer_and_random_gener_of_acces_autom} presented
also a random generator of \ICDFA's based on Boltzmann Samplers,
recently introduced by Duchon \emph{et
  al.}~\cite{duchon04:_boltz_sampl_for_random_gener}. However the
sampler is uniform for partitions of a set with $kn$ elements into $n$
nonempty subsets (not for the universe of automata). These partitions
are related with string representations that verify only rule
\ref{eq:r1}. Based on the work here presented, it would be interesting
to study a better approximation, that would satisfy rule~\ref{eq:r2}.
\section{Acknowledgments}
We thank the anonymous referees for their
comments that helped to improve this~paper.

\end{document}